\documentclass[10pt,twocolumn,twoside]{IEEEtran}

\usepackage[utf8x]{inputenc}
\usepackage[cmex10]{amsmath}
\usepackage{amsfonts}
\usepackage{amsthm}
\usepackage{algorithmic}
\usepackage{amsmath}

\usepackage{amsfonts}
\usepackage{amssymb}
\usepackage{graphicx}
\usepackage{url}
\usepackage{hyperref}
\usepackage[usenames,dvipsnames]{color}
\usepackage{caption}
\usepackage{subcaption}

\usepackage{tabularx}

 \usepackage{setspace}
\usepackage{enumerate}

\newtheorem{lemma}{Lemma}
\newtheorem{theorem}{Theorem}
\newtheorem{corollary}{Corollary}
\newtheorem*{rem*}{Remark}

\theoremstyle{definition}
\newtheorem{defn}{Definition}

\newtheorem{assumption}{Assumption}

\newcommand{\revA}{Black}
\newcommand{\revB}{Black}

\begin{document}

\title{On the necessity of symmetric positional coupling for string stability}

\author{Dan~Martinec, 
        Ivo~Herman, 
        and~Michael~\v{S}ebek
\thanks{All authors are with the Faculty of Electrical Engineering, Czech
Technical University in Prague, Department of Control Engineering, Karlovo
namesti 13, 121 35 Prague, Czech Republic.
E-mail address:
\{martinec.dan, ivo.herman, sebekm1\}@fel.cvut.cz }
\thanks{The research
was supported by the Grant Agency of the Czech Republic within the projects GACR
13-06894S (I.~H.).}}

\maketitle

\begin{abstract}
We consider a distributed system with identical agents, constant-spacing policy and asymmetric bidirectional control, where the asymmetry is due to different controllers, which we describe by transfer functions. By applying the wave transfer function approach, it is shown that, if there are two integrators in the dynamics of agents, then the positional coupling must be symmetric, otherwise the system is locally string unstable. This finding holds also for a distributed system with a {\color{\revA}generalized path-graph} interaction topology due to the local nature of the wave transfer function. The main advantage of the transfer function approach is that it allows us to analyse the bidirectional control with an arbitrary complex asymmetry in the controllers, for instance, the control with symmetric positional but asymmetric velocity couplings.


\end{abstract}


\IEEEpeerreviewmaketitle


\section{Introduction}



Although each agent in a distributed system is usually well designed and asymptotically stable, the interaction between agents may trigger undesirable phenomena such as string instability. There are several slightly different definitions of the string instability, see \cite{Eyre1998a}, \cite{Ploeg2014} or \cite{Swaroop1996}. They all describe how the disturbance acting on an agent amplifies as it propagates in the system. Similar analytical measures of system performance are harmonic stability \cite{Tangerman2012}, flock stability \cite{Cantos2013} and coherence \cite{Bamieh2012b}.

One of the most studied distributed system is a vehicular platoon, where the interaction topology is a path graph. Each agent of such a system, except for the first and last ones, measures the distance, i.e. its relative position, to the nearest neighbours. There are several spacing policies that can be used to control a vehicular platoon. The two most popular are the constant-spacing policy, where the goal is to keep a constant distance between the agents, and the time-headway policy, where the agents are forced to keep a constant time gap between them. In this paper, we consider the former policy. For such a case, we require two integrators to be present in the open-loop model of each agent so that the agent can track the leader travelling with a constant velocity with the zero steady-state error.

Considering two integrators in the model makes an investigation of the string stability for the constant-spacing policy challenging. It was shown in \cite{Seiler2004a} that the string instability is unavoidable for the agents with two integrators under an unidirectional interaction, therefore, an asymmetric bidirectional scheme was introduced. Later, it was shown in \cite{Hao2011} that the identical asymmetry for all states used for coupling causes a nonzero lower bound on the distributed-system eigenvalues, which guarantees the controllability of a system with even a large number of agents, see \cite{Barooah2009}. \cite{Tangerman2012} and \cite{Herman2014c} show that the disadvantage of such an asymmetric bidirectional control of agents with two integrators is that the system is harmonically unstable, meaning that the $\mathcal{H}_{\infty}$ norm of the transfer functions between the agents scales exponentially with the number of agents in the system.



Recently, papers \cite{Hao2012c} and \cite{Cantos2013} introduce a novel type of asymmetric bidirectional control by assuming nonequal asymmetries between the output states. They showed that different couplings between the positions and velocities in the double integrator system can be beneficial for decreasing the transient and overshoots of the system response. The latter paper also suggests that the symmetry in the positional coupling is necessary for the asymptotic and flock stabilities of an oscillator array. The reasoning of both papers were based on mathematical simulations and reasonable conjectures, which raise the following questions. Can the 'symmetry' condition be generalized for more complex agent dynamics? Is the symmetric coupling necessary for other types of graphs than a path graph? Answering these questions is the main aim of this paper.


It is difficult to answer these questions using the traditional Laplacian approach, e.g. \cite{Olfati-Saber2007}. We therefore use the approach, where the response of the system is decomposed into two travelling waves which are described by an irrational transfer function. The travelling-wave approach originates in the analysis and modelling of the flexible mechanical structures, see \cite{vaughan_application_1968} and \cite{Flotow1985}. The concept was revisited in a series of papers by O'Connor in \cite{OConnor2006} and \cite{OConnor2007}, under the term \emph{wave-based control}, for a control of lumped flexible systems. Paralely, the wave concept was considered for the control of continuous flexible structures in \cite{Halevi2005} and \cite{Halevi2006}. Recently, the travelling-wave approach was applied on distributed control in \cite{Martinec2014a}, \cite{Martinec2014c} and \cite{Martinec2015b}.

In this paper, we adopt the wave approach from \cite{Martinec2014a} and generalize it for a homogenous-asymmetric path graph, which represents a distributed system where all agents are identical but the coupling between them is asymmetric. Unlike the traditional Laplacian approach, the wave approach allows us to describe how the information is locally propagated from an agent to its immediate neighbours. By analyzing this local behaviour, we can study the performance of a distributed system, for instance, the string (in)stability. Moreover, the wave approach allows the treatment of arbitrary asymmetry in the controllers, for instance, different positional and velocity couplings. We show that symmetric coupling between the agent positions, represented by the identical DC gains of the controllers, is necessary for the string stability. This result holds for a constant-spacing policy with an arbitrary agent model, which is a complementary result to prior findings about the string stability of an asymmetric bidirectional control.


{\color{\revA} Further, the paper introduces a new type of string stability, the so-called local string stability, which simplifies the performance analysis of a large distributed system. The definition of the local string stability, similarly as the definition of string stability for instance in \cite{Ploeg2014}, \cite{Seiler2004a} or \cite{Eyre1998a}, captures whether the disturbance acting on an agent amplifies as it propagates through the path-graph system. The main difference is that the local string stability disregards the boundary conditions on the path-graph ends. The advantage is that it allows us to assess a system where the path graph is only a part of the interaction topology, since the analysis of the local string stability does not distinguish whether the boundary condition is caused by the agent on the edge of the system or another, more complex, part of the system. The disadvantage is that the local string stability gives only a sufficient condition for the string stability as is shown in Lemma~\ref{lemma:connection_to_string_stability}.

}




\section{Mathematical preliminaries}

We consider a formation of identical agents with a path-graph interaction topology, for instance, a platoon of vehicles on a highway. The goal of the formation is to move along a line with equal distances between the agents.

The dynamics of agents is described by a linear single-input-single-output model. The output of the model is the position of the agent, $X_{n}(s)$, described as
\begin{align}
  X_n(s) = P(s) U_n(s),
\end{align}
where $n$ denotes the index of $n$th agent, $P(s)$ is the transfer function of the model and $U_n(s)$ is the input to the agent generated by the local controllers onboard the agent. The goal of the controllers is to equalize the distances to the immediate neighbours. Each agent has two controllers $C_{\text{f}}(s)$ and $C_{\text{r}}(s)$ for controlling the front and rear distances of the agent. We describe the controllers by transfer functions, which allows the representation of arbitrary couplings between the agents. In other words, the controllers may be of an arbitrary order and structure. We consider that each agent has the same set of controllers but the two controllers may be different, i.e. $C_{\text{f}}(s) \neq C_{\text{r}}(s)$. Then
\begin{align}
  U_n(s) = &C_{\text{f}}(s)\left(X_{n-1}(s) - X_{n}(s)\right)\nonumber\\
   &+ C_{\text{r}}(s) (X_{n+1}(s) - X_{n}(s)).
\end{align}
The resulting model of the $n$th agent is shown in Fig.~\ref{fig:asymmetric_agent} and described as
\begin{align}
  X_{n}(s) = &M_{\text{f}}(s)(X_{n-1}(s)-X_{n}(s)) \nonumber\\
  &+ M_{\text{r}}(s)(X_{n+1}(s)-X_{n}(s)),\label{eq:eq1}
\end{align}
where $M_{\text{f}}(s) = C_{\text{f}}(s)P(s)$ and $M_{\text{r}}(s) = C_{\text{r}}(s)P(s)$.

\begin{figure}[ht]
 \centering
  \includegraphics[width=0.33\textwidth]{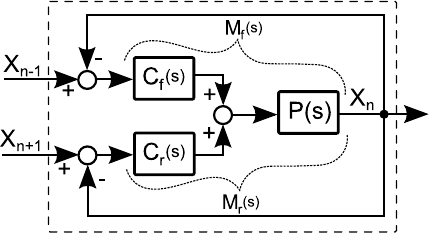}
  \caption{The model of $n$th agent.}
  \label{fig:asymmetric_agent}
\end{figure}

The first agent ($n=0$), the so-called leader, is externally controlled and serves as a reference signal for the distributed system. The rear-end agent ($n=N$) of the path graph is described as
\begin{align}
  X_{N}(s) = M_{\text{f}}(s)(X_{N-1}(s)-X_{N}(s)).\label{eq:eq2}
\end{align}

\begin{assumption}\label{assum:1}
Throughout the paper, we assume that
  \begin{enumerate}[(a)]
    \item $M_{\text{f}}(s)$ and $M_{\text{r}}(s)$ have the same number of $p$ integrators,
    \item $M_{\text{f}}(s)$ and $M_{\text{r}}(s)$ are proper,
    \item $M_{\text{f}}(s)$ and $M_{\text{r}}(s)$ have no CRHP (closed-right half plane) zeros and no CRPH poles except of $p$ poles in the origin.
  \end{enumerate}
\end{assumption}

It is convenient to express $M_{\text{f}}(s)$ and $M_{\text{r}}(s)$ as
\begin{align}
  M_{\text{f}}(s) &= \frac{1}{s^p} \frac{n_{\text{f}}(s)}{d_{\text{f}}(s)} = \frac{1}{s^p} \frac{\sum_{k=0}^{L_{\text{f}}}n_{\text{f},k} s^k}{\sum_{k=0}^{K_{\text{f}}} d_{\text{f},k}s^k},\label{eq:as_in_1a} \\
  M_{\text{r}}(s) &= \frac{1}{s^p} \frac{n_{\text{r}}(s)}{d_{\text{r}}(s)} = \frac{1}{s^p} \frac{\sum_{k=0}^{L_{\text{r}}}n_{\text{r},k} s^k}{\sum_{k=0}^{K_{\text{r}}}d_{\text{r},k}s^k},\label{eq:as_in_1b}
\end{align}
where $K_{\text{f}}$, $L_{\text{f}}$, $K_{\text{r}}$ and $L_{\text{r}}$ are the orders of polynomials $n_{\text{f}}(s)$, $d_{\text{f}}(s)$, $n_{\text{r}}(s)$ and $d_{\text{r}}(s)$, respectively, and $n_{\text{f},k}$, $d_{\text{f},k}$, $n_{\text{r},k}$ and $d_{\text{r},k}$ are their coefficients. Without loss of generality we assume $n_{\text{f},0} \neq 0$, $n_{\text{r},0} \neq 0$ and $d_{\text{f},0} = d_{\text{r},0} = 1$.


The traditional asymmetric bidirectional control, see \cite{Barooah2009} or \cite{Tangerman2012}, assumes that $M_{\text{f}}(s) = \mu M_{\text{r}}(s)$, where $\mu$ is a constant gain. We allow the asymmetry to be more general than scaling and focus on the relation between the $k$th coefficients of (\ref{eq:as_in_1a}) and (\ref{eq:as_in_1b}). 

\begin{defn}\label{def:pos_coupling}
We say that the distributed system has \emph{symmetric positional coupling} if the open-loop model of an agent satisfies
  \begin{align}
  \frac{n_{\text{f},0}}{d_{\text{f},0}} = \frac{n_{\text{r},0}}{d_{\text{r},0}}.
  \end{align}
\end{defn}
In other words, the positional coupling is symmetric if the DC gain of $M_{\text{f}}(s)/M_{\text{r}}(s)$ is equal to one. Similarly, the velocity coupling is symmetric if $n_{\text{f},1}/d_{\text{f},1} = n_{\text{r},1}/d_{\text{r},1}$.

\begin{defn}\label{def:path_graph}
{\color{\revA}
We say that the interaction topology of $N+1$ agents is a \emph{path graph}, if $N \geq 3$, the $n$th agent, $n \in (1,N-1)$, is described by (\ref{eq:eq1}), the $0$th agent is externally controlled and the $N$th agent is described by (\ref{eq:eq2}).}
\end{defn}



\section{Wave transfer function for asymmetric bidirectional connection}

\subsection{Introduction of the wave approach}

The bidirectional property of locally controlled agents causes that any change in the position of the leader is propagated through the distributed system as a \emph{wave}. When the wave reaches the rear-end agent, it reflects and propagates back to the leader, where it reflects again. This section describes the propagation of this wave.

The basic idea is to describe the position of the $n$th agent in a distributed system with a path-graph topology by two components, $A_n(s)$ and $B_n(s)$, which represent two waves propagating along a distributed system in the forward and backward directions, respectively. The mathematical model of a distributed system with a path-graph topology is shown in Fig.~\ref{fig:AWTF_graph_single} and described as
\begin{align}
  X_n(s) &= A_n(s) + B_n(s),\label{eq:pos_decomp}\\
  A_{n+1}(s) &= G_{+}(s)A_n(s),\label{eq:anp_s}\\
  B_n(s) &= G_-(s)B_{n+1}(s)\label{eq:bnp_s},
\end{align}
where $n \in \{1, 2,...,N −1\}$, $G_{+}(s)$ and $G_{-}(s)$ are \emph{asymmetric wave transfer functions} (AWTFs), which describe how the wave propagates in the system in the forward, (\ref{eq:anp_s}), and backward, (\ref{eq:bnp_s}), directions, respectively.
\begin{figure}[ht]
 \centering
  \includegraphics[width=0.35\textwidth]{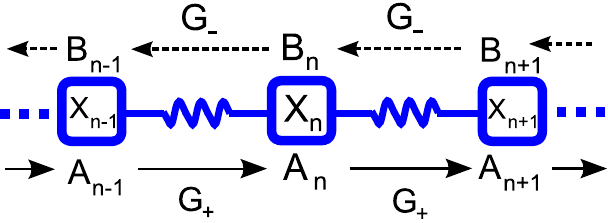}
  \caption{
  Scheme of waves travelling in a distributed system with a path-graph topology. The squares stand for agents and springs illustrate the virtual connections between the agents created by the controllers. Note that all the agents are identical.}
  \label{fig:AWTF_graph_single}
\end{figure}

\begin{lemma}\label{lem:AWTF}
AWTFs $G_{+}(s)$ and $G_{-}(s)$ in (\ref{eq:pos_decomp})-(\ref{eq:bnp_s}) are given by
\begin{align}
  G_{+}(s) &= \frac{1}{2}\beta(s)-\frac{1}{2}\sqrt{\beta^2(s)-4\frac{M_{\textnormal{f}}(s)} {M_{\textnormal{r}}(s)}}, \label{eq:AWTF_1}\\
  G_{-}(s) &= \frac{1}{2}\alpha(s)-\frac{1}{2}\sqrt{\alpha^2(s)-4\frac{M_{\textnormal{r}}(s)} {M_{\textnormal{f}}(s)}} \label{eq:AWTF_2},
\end{align}
where $M_{\textnormal{f}}(s)$ and $M_{\textnormal{r}}(s)$ define the system in (\ref{eq:eq1}) and
\begin{align}
\alpha(s) = \frac{1+M_{\textnormal{f}}(s)+M_{\textnormal{r}}(s)}{M_{\textnormal{f}}(s)}, \;\; \beta(s) = \frac{1+M_{\textnormal{f}}(s)+M_{\textnormal{r}}(s)}{M_{\textnormal{r}}(s)}.\label{eq:AWTF_3}
\end{align}
\end{lemma}
\begin{proof}
The proof is based on the same approach as in Section~3.2~of~\cite{Martinec2014a} or also in Section~3.1~of~\cite{OConnor2006}. We note that the Laplace variable `s' is dropped in the following notation.

The substitution of (\ref{eq:pos_decomp})-(\ref{eq:bnp_s}) into (\ref{eq:eq1}) yields
\begin{align}
  A_n + B_n &= M_{\text{f}}\left((G_+^{-1}A_{n}+G_-B_{n})-(A_n+B_n)\right)\nonumber\\
  & + M_{\text{r}}\left((G_+A_{n}+G_-^{-1}B_{n})-(A_n+B_n)\right).\label{eq:lem1_AWTF_pf1}
\end{align}
This equation can be decomposed into $A$ and $B$ parts as
\begin{align}
  1 &= M_{\text{f}}G_{+}^{-1}-M_{\text{f}}+M_{\text{r}}G_+ - M_{\text{r}},\label{eq:lem1_AWTF_pf2a}\\
  1 &= M_{\text{f}}G_{-} -M_{\text{f}} + M_{\text{r}}G_-^{-1} - M_{\text{r}}.\label{eq:lem1_AWTF_pf2b}
\end{align}
We rearrange it and get
\begin{align}
  &G_+^{2}-\beta G_+ + \frac{M_{\text{f}}}{M_{\text{r}}} = 0,\label{eq:lem1_AWTF_pf3a}\\
  &G_-^{2}-\alpha G_- + \frac{M_{\text{r}}}{M_{\text{f}}} = 0,\label{eq:lem1_AWTF_pf3b}
\end{align}
where $\alpha$ and $\beta$ are from (\ref{eq:AWTF_3}). The solutions of the quadratic equations are given as
\begin{align}
  G_{+}(s)_{1,2} &= \frac{1}{2}\beta(s) \pm \frac{1}{2}\sqrt{\beta^2(s)-4\frac{M_{\textnormal{f}}(s)} {M_{\textnormal{r}}(s)}},\label{eq:lem1_AWTF_pf4a} \\
  G_{-}(s)_{1,2} &= \frac{1}{2}\alpha(s) \pm \frac{1}{2}\sqrt{\alpha^2(s)-4\frac{M_{\textnormal{r}}(s)} {M_{\textnormal{f}}(s)}}.\label{eq:lem1_AWTF_pf4b}
\end{align}

We have specified that $G_+(s)$ describes the wave propagating along the system in the forward direction, i.e. from $n$th to $(n+1)$th agent. The propagation of the wave is causal, therefore, the transfer function describing this phenomenon must be either proper or strictly proper. We will show that the transfer functions $G_{+}(s)$ with a plus sign in front of the second term in (\ref{eq:lem1_AWTF_pf4a}) is not proper.

The definition of a proper irrational transfer function is given by Definition~B.1~in~\cite{Curtain2009}, which states: \emph{The function $G$ is proper if for sufficiently large $\rho$
\begin{align}
  \sup_{\text{Re }s\geq 0 \cap |s|>\rho} |G(s)| < \infty.
\end{align}}

Due to Assumption~\ref{assum:1}, the norm of (\ref{eq:lem1_AWTF_pf4a}) can be unbounded only for $s \rightarrow \infty$. In addition, also $\lim_{s\rightarrow \infty} \beta(s) = \infty$, $\lim_{s \rightarrow \infty} \sqrt{\beta^2-4M_{\text{f}}/M_{\text{r}}} = \lim_{s \rightarrow \infty} \beta$. Therefore {\color{\revB} (\ref{eq:AWTF_1}) and (\ref{eq:AWTF_2}) are the only proper solutions.}


\end{proof}

\begin{rem*}
We can see that the AWTFs are linear, however, it may be surprising that they are irrational. {\color{\revB}The independence of $A_n(s)$ and $B_n(s)$ on each other is possible only for the system with infinite number of agents. The infinite dimensionality of the system then makes the transfer functions to be irrational.


A system with a finite number of identical agents and the path-graph topology has the leader and the rear-end agent. They act as boundaries for the travelling waves and cause their reflections.} In basic wave physics, the boundary is assumed to satisfy the spatial causality, that is, the boundary condition does not affect the wave travelling towards it. In other words, (\ref{eq:pos_decomp})-(\ref{eq:bnp_s}) hold regardless of the topological distance and dynamics of the rear-end agent. Therefore, we can apply (\ref{eq:pos_decomp})-(\ref{eq:bnp_s}) to describe the travelling waves even in a system with a finite number of agents, although, these relations are valid only for the agents that are not placed on the boundary. The boundary agents that causes the reflection of the waves must be treated separately, see Lemma~\ref{lem:refl}. In general, any agent that is not described by (\ref{eq:eq1}) represents a boundary for the travelling wave. This applies also to an agent that has more than two neighbours, see \cite{Martinec2015b}. Therefore, the travelling-wave decomposition, (\ref{eq:pos_decomp})-(\ref{eq:AWTF_2}) is valid even for a {\color{\revA}generalized path graph, see Definition~\ref{def:generalized_path_graph}}.

\end{rem*}

We note that the reflections of the wave from the leader and the rear-end agent described by the following Lemma are not used in the derivation of the main result of this paper. However, we feel obliged to derive them to fully cover the issue of waves in asymmetric bidirectional control. Moreover, we use the reflections for numerical verification of the proposed AWTF approach.

\begin{lemma}\label{lem:refl}
  The reflection from the leader and the rear-end agent in the path graph is described by the transfer function $T_{1}(s) = A_1(s)/B_1(s)$ and $T_{N}(s) = B_{N}(s)/A_N(s)$, respectively. The transfer functions are given as
  \begin{align}
    T_{1}(s) &= \frac{A_1(s)}{B_1(s)} = -G_+(s) G_-(s),\label{eq:refl1a} \\
    T_{N}(s) &= \frac{B_N(s)}{A_N(s)} = G_-(s)\frac{G_+(s) -1}{G_-(s) -1},\label{eq:refl1b}
  \end{align}
  respectively.
\end{lemma}
\begin{proof}
The position $X_1(s)$ in (\ref{eq:eq1}) can be rewritten using (\ref{eq:AWTF_3}) as
\begin{align}
  X_1 = \frac{1}{\alpha}X_0 + \frac{1}{\beta}X_2.
\end{align}
Substituting for $X_1 = A_1+B_1$ and $X_2 = G_+A_1 + G_-^{-1}B_1$ from (\ref{eq:anp_s}) and (\ref{eq:bnp_s}), it yields
\begin{align}
  A_1 &= \frac{1}{\alpha} \frac{1}{1-\dfrac{1}{\beta}G_+} X_0 + \frac{1}{\beta}\frac{\alpha}{\alpha}\frac{G_-^{-1} -\beta}{1-\dfrac{1}{\beta}G_+}B_1.\label{eq:pf_Lem2_4}
\end{align}

The last expression can be further simplified by the following arrangements. First, from (\ref{eq:AWTF_3}) we have
\begin{align}
  \frac{\beta}{\alpha} = \frac{M_{\text{f}}}{M_{\text{r}}}.\label{eq:pf_Lem2_5}
\end{align}
Further, by (\ref{eq:AWTF_1}) and (\ref{eq:pf_Lem2_5}), it can be shown that
\begin{align}
  G_{+}^{-1} &= \frac{M_{\text{r}}}{M_{\text{f}}}\left(\frac{1}{2}\beta+\frac{1}{2}\sqrt{\beta^2- 4\frac{M_{\text{f}}}{M_{\text{r}}}}\right) \nonumber\\
  &= \frac{\alpha}{\beta}(\beta-G_{+}) = \alpha\left(1-\frac{1}{\beta}G_{+}\right).\label{eq:pf_Lem2_6}
\end{align}
Likewise,
\begin{align}
  G_{-}^{-1} &=  \frac{\beta}{\alpha}(\alpha-G_{-}).\label{eq:pf_Lem2_7}
\end{align}
By rearranging (\ref{eq:pf_Lem2_7}), it gives
\begin{align}
  G_- = \alpha - \frac{\alpha}{\beta}G_{-}^{-1}.\label{eq:pf_Lem2_8}
\end{align}
Substituting (\ref{eq:pf_Lem2_5}), (\ref{eq:pf_Lem2_6}) and (\ref{eq:pf_Lem2_8}) into (\ref{eq:pf_Lem2_4}) gives
\begin{align}
  A_1 &= G_+ X_0 - G_+ G_- B_1 = G_+ X_0 + T_1 B_1.
\end{align}

Now, we derive the reflection relation for the rear-end agent. Substituting (\ref{eq:anp_s}), (\ref{eq:bnp_s}) and (\ref{eq:pos_decomp}) into (\ref{eq:eq2}) gives
\begin{align}
  A_N+B_N &= M_{\text{f}}(G_+^{-1}A_N+G_- B_N - A_N - B_N).
\end{align}
By rearranging, it gives
\begin{align}
  B_N = \frac{1+M_{\text{f}}-M_{\text{f}}G_{+}^{-1}}{M_{\text{f}}G_- -M_{\text{f}}-1} A_N.\label{eq:pf_Lem2_11}
\end{align}
By (\ref{eq:pf_Lem2_6}) and (\ref{eq:pf_Lem2_8}), we have $M_{\text{f}}G_+^{-1} = -M_{\text{r}}G_{+} + (M_{\text{r}}+M_{\text{f}}+1)$ and $M_{\text{f}}G_- = -M_{\text{r}}G_{-}^{-1}+(M_{\text{r}}+M_{\text{f}}+1)$. Substituting these into (\ref{eq:pf_Lem2_11}) results in
\begin{align}
  B_N = \frac{M_{\text{r}}(G_+ -1)}{M_{\text{r}}(1-G_-^{-1})}A_N = G_- \frac{G_+ -1}{G_- -1}A_N.
\end{align}

\end{proof}

\subsection{Discussion of the wave approach}\label{sec:wave_approach_discussion}

Let us consider a system of 20 agents with path-graph topology and $M_{\text{f}}(s)$ and $M_{\text{r}}(s)$ defined later in the paper (Section~\ref{sec:simulations}). The transfer function $T_{0,10}(s) = X_{10}(s)/X_0(s)$ can be found by the traditional state-space approach, or by recursive application of (\ref{eq:eq1}) and (\ref{eq:eq2}). This transfer function takes into consideration the interactions among all the other agents and the effect of the boundaries. Therefore, it is an `overall' description of the system, which is well suited for investigating the asymptotic stability of the system.

The `overall' transfer function $T_{0,10}(s)$ can alternatively be found by ratio $(A_{10}(s)+B_{10}(s))/(X_0(s))$. However, this requires to consider the reflections on the leader and the rear-agent, described by (\ref{eq:refl1a}) and (\ref{eq:refl1b}), as follows
\begin{align}
  A_{10} &= G_+^{10}X_0+ T_N G_+^{19} T_1 G_-^{19}A_{10},\label{eq:com_1a}\\
  B_{10} &= G_-^{10}T_N G_+^{20}X_0+ T_N G_+^{19} T_1 G_-^{19}B_{10}.\label{eq:com_1b}
\end{align}
The first term on the right-hand side describes the wave traveling to the agent due to a change of $X_0(s)$. The second term describes the wave returning back to the agent due to the reflections on the leader and the rear-end agent. Equations (\ref{eq:com_1a}) and (\ref{eq:com_1b}) can be simplified as
\begin{align}
  X_{10}(s) = \frac{G_+^{10}(s) + G_-^{10}(s)T_N(s)G_+^{20}(s)}{1-T_N(s)G_-^{19}(s)T_1(s)G_+^{19}(s)}X_0(s).\label{eq:com_2}
\end{align}
It can be shown that the transfer function in (\ref{eq:com_2}) is rational and equal to $T_{0,10}(s)$. We can see that considering the reflections in the system is rather cumbersome.

On the other hand, the most important aspect of the wave approach is that it describes the system from the local point-of-view and allows us to decompose the output of the agents into two travelling waves. It takes a certain time for the wave to propagate in the system, therefore, we can approximate the output of the first agent as
\begin{align}
  X_1(s) \approx A_1(s) \approx G_+(s)X_0(s).\label{eq:com_3}
\end{align}
Similarly, for the second agent,
\begin{align}
  X_2(s) \approx A_2(s) \approx G_+^2(s)X_0(s),\label{eq:com_4}
\end{align}
etc. The approximation gives the exact result in the time-domain until the wave propagates to the last agent, reflects and travels back to the $n$th agent. The important aspect is that the approximation is analytic allowing the analysis of its properties. Based on that, we can infer properties of the multi-agent system.

We should emphasize that the `overall' description in (\ref{eq:com_2}) holds only for one particular system, that is, for 20 agents with the path-graph topology. However, the `local' description, that is $A_{n+1}(s) = G_{+}(s)A_{n}(s)$ and $B_{n}(s) = G_{-}(s)B_{n+1}(s)$, holds for arbitrary graph that contains a path graph due to the spatial causality of the boundary (see the discussion in Remark after Lemma~\ref{lem:AWTF}).

\subsection{Properties of AWTFs}

To be able to track the leader travelling at a constant velocity with the zero steady-state error, we require two integrators to be present in the model of each agent. The DC gains of the AWTFs, in this case, are limited to one as the following Lemma describes.
\begin{lemma}\label{lemma:DC_gains}
  If there is at least one integrator in $M_{\textnormal{f}}(s)$ and $M_{\textnormal{r}}(s)$, defined by (\ref{eq:as_in_1a})-(\ref{eq:as_in_1b}), then the DC gains of the AWTFs given by (\ref{eq:AWTF_1}) and (\ref{eq:AWTF_2}) are
  \begin{align}
    \lim_{s \rightarrow 0} G_+(s) &= \kappa, \;\;\; \lim_{s \rightarrow 0} G_-(s) = 1, \;\;\;\;\;\;\; \textnormal{ if } 0< \kappa < 1,\label{eq:DC_AWTF1a}\\
     \lim_{s \rightarrow 0} G_+(s) &= 1, \;\;\; \lim_{s \rightarrow 0} G_-(s) = 1/\kappa, \;\;\; \textnormal{ if } \kappa \geq 1,\label{eq:DC_AWTF1b}
  \end{align}
  where
  \begin{align}
  \kappa = \lim_{s\rightarrow 0} \frac{M_{\textnormal{f}}(s)}{M_{\textnormal{r}}(s)} = \frac{n_{\textnormal{f},0}}{n_{\textnormal{r},0}}.
  \end{align}
\end{lemma}
\begin{proof}
  First, we prove the DC gain of $G_+$. Since there is at least one integrator in $M_{\text{r}}(s)$, then the limit of $\beta(s)$ given by (\ref{eq:AWTF_3}) is
  \begin{align}
    \lim_{s \rightarrow 0}\beta(s) = \lim_{s \rightarrow 0} \left(1+\frac{1}{M_{\text{r}}(s)}+\frac{M_{\text{f}}(s)}{M_{\text{r}}(s)}\right) = 1+\kappa.\label{eq:DC_AWTF3}
  \end{align}
  Substituting from (\ref{eq:DC_AWTF3}) into (\ref{eq:AWTF_1}) gives
  \begin{align}
  \lim_{s \rightarrow 0} G_+(s) &= \frac{1}{2}\left(1+\kappa-\sqrt{(1+\kappa)^2-4\kappa}\right)\nonumber\\
  &= \frac{1}{2}(1+\kappa -|1-\kappa|).\label{eq:lim_Gplus}
  \end{align}
  The proof of the DC gain of $G_-$ is similar.
\end{proof}


\begin{defn}
  We say that the asymmetric wave transfer function $T(s)$ is \emph{asymptotically stable} if it is analytic in the right-half plane and $||T||_{\infty} < \infty$, where $||T||_{\infty} = \sup_{Re(s) >0} |T(s)|$.
\end{defn}
This definition follows the definition of the stability of a linear system by Theorem~A.2~of~\cite{Curtain2009}.

\begin{theorem}\label{lem:AWTF_stab}
If $M_{\textnormal{f}}(s)$ and $M_{\textnormal{r}}(s)$ defined by (\ref{eq:as_in_1a})-(\ref{eq:as_in_1b}) satisfy Assumption~\ref{assum:1} and if the Nyquist plot of
 \begin{align}
  T_{G}(s) = (M_{\textnormal{f}}(s)-M_{\textnormal{r}}(s))^2+ 2M_{\textnormal{f}}(s)+2M_{\textnormal{r}}(s) +1
\end{align}
does not intersect the non-positive real axis, then the AWTFs {\color{\revB}given by (\ref{eq:AWTF_1}) and (\ref{eq:AWTF_2})} are asymptotically stable.
\end{theorem}
\begin{proof}
We have shown that the norms of $G_+(s)$ and $G_-(s)$ are bounded in the proof of Lemma~\ref{lem:AWTF}, hence, we focus on their analyticity. We use the result of the complex function analysis, which states that the square root function $f(z)= \sqrt{z}$ is analytic everywhere except for the non-positive real axis (e.g., \cite{stein2010complex}). The second term of $G_+(s)$ is
\begin{align}
f_{2,+}(s) &=  \frac{1}{2}\sqrt{\beta^2(s) - 4\frac{M_{\text{f}}(s)}{M_\text{r}(s)}} = \frac{1}{2}\sqrt{\frac{1+M_{\text{s}}(s)}{M_{\text{r}}^2(s)}},\label{eq:pf_stab4}
\end{align}
where $M_{\text{s}} = M_{\text{r}}^2+M_{\text{f}}^2 + 2M_{\text{r}} + 2M_{\text{f}} - 2M_{\text{r}}M_{\text{f}} = (M_{\text{r}}-M_{\text{f}})^2+2(M_{\text{r}}+M_{\text{f}})$. We apply the same analysis as in Figure 1.9 of Section 1.2 of \cite{Kelly2006}. Term $\sqrt{1/M_{\text{r}}}$ is analytic everywhere except for the non-positive real axis, where it has a branch cut. The non-analyticity is caused by functional discontinuity, which is, in this case, only a sign change. Due to that, the overlapping branch cuts of $\sqrt{1/M_{\text{r}}\cdot 1/M_{\text{r}}}$ cancel each other, which means that $\sqrt{1/M_{\text{r}}^2}$ is continuous and analytic even on the non-positive real axis. Then $f_{2,+}$ is analytic if and only if $\sqrt{1+M_{\text{s}}}$ is analytic. Hence, if the Nyquist plot of $1+M_{\text{s}}$ does not intersect the non-positive real axis, then $f_{2,+}$ is analytic.

The first term of the AWTFs, $\alpha/2$ and $\beta/2$, are rational transfer functions. A rational function is analytic in the ORHP (open-right half plane) if and only if it has no singularities, in this case ORHP zeros and ORHP poles of $M_{\text{f}}$ and $M_{\text{r}}$. Therefore, if $M_{\text{f}}$ and $M_{\text{r}}$ have no ORHP zeros, nor ORHP poles, then $\alpha/2$ and $\beta/2$ are analytic.

Observe that
\begin{align}
  G_+(s) = \frac{1}{2}\beta(s)-\frac{1}{2}\sqrt{\frac{1+M_{\text{s}}(s)}{M_r^2(s)}},
\end{align}
is an analytic function since the Nyquist plot of $(1+M_{\text{s}}(s))$ does not intersect the non-positive real axis and $\beta(s)$ and $\alpha(s)$ do not have ORHP poles due to the condition that $M_{\text{f}}(s)$ and $M_{\text{r}}(s)$ have no CRHP zeros, nor CRHP poles.

\end{proof}

\begin{theorem}\label{thm:two_integrators_Mf_Mr}
If the AWTFs given by (\ref{eq:AWTF_1}) and (\ref{eq:AWTF_2}) are asymptotically stable, there are two integrators in $M_{\textnormal{f}}(s)$ and $M_{\textnormal{r}}(s)$, given by (\ref{eq:as_in_1a})-(\ref{eq:as_in_1b}), and
  \begin{align}
  n_{\textnormal{f},0} &\neq n_{\textnormal{r},0},\label{eq:as_in_2a}\\
  d_{\textnormal{f},0} &= d_{\textnormal{r},0} = 1,\label{eq:as_in_2b}\\
  n_{\textnormal{f},0} >&\, 0,\; n_{\textnormal{r},0}>0,\label{eq:as_in_2c}
  \end{align}
  then either $||G_+(s)||_{\infty} > 1$ or $||G_-(s)||_{\infty} > 1$.
\end{theorem}
\begin{proof}
First, we prove that $||G_+||_{\infty}>1$ if $\kappa>1$. By $\omega_0$ we denote a frequency that is close to $0$ and evaluate the real and imaginary parts of the individual transfer functions as
\begin{align}
  x_1 + \jmath y_1 = \frac{M_{\text{f}}(\jmath\omega_0)}{M_{\text{r}}(\jmath\omega_0)},\;\; x_2 + \jmath y_2 = \frac{1}{M_{\text{r}}(\jmath\omega_0)},\label{eq:pfV_18}
\end{align}
$x = x_1+x_2$ and $y = y_1+y_2$. The Taylor series expansions of (\ref{eq:pfV_18}) evaluated at zero are
\begin{align}
x_1(\omega_0) &= k_{x,1} - k_{x,2}\omega_0^2 + k_{x,3}\omega_0^4 - ...,\label{eq:pfV_19a}\\
y_1(\omega_0) &= k_{y,1}\omega_0 - k_{y,2}\omega_0^3 + k_{y,3}\omega_0^5 - ...,\label{eq:pfV_19b}\\
x_2(\omega_0) &= -l_{x,1}\omega_0^2 + l_{x,2} \omega_0^4 - l_{x,3}\omega_0^6 + ...,\label{eq:pfV_19c}\\
y_2(\omega_0) &= - l_{y,1}\omega_0 ^3 + l_{y,2}\omega_0^5 - l_{y,3}\omega_0^7 +  ...,\label{eq:pfV_19d}
\end{align}
where we assume that $M_{\text{r}}$ has two integrators. We note that $k_{x,1} = \kappa$ and $l_{x,1} = 1 / n_{\text{r},0}$. The other coefficients, $k_{x,2}$, $k_{x,3}$, etc., obtained by Taylor series are not important due to limit $\omega_0 \rightarrow 0$ as we show later in the proof.

Substituting (\ref{eq:pfV_18}) into (\ref{eq:AWTF_1}) gives the real part of $G_+(\jmath\omega_0)$ as
\begin{align}
    \operatorname{Re}\{G_+(\jmath\omega_0)\} =& \frac{1}{2}\left(1+x\right)
    -\frac{1}{2}\sqrt{\frac{|z|+\operatorname{Re}\{z\}}{2}},\label{eq:pfV_20}
\end{align}
where $\operatorname{Re}\{\sqrt{z}\} = \sqrt{|z|/2+\operatorname{Re}\{z\}/2}$, see e.g. Section 3.7.27 in \cite{Abramowitz1964}, and
\begin{align}
  z &= \beta^2(\jmath \omega_0)-4\frac{M_{\text{f}}(\jmath \omega_0)}{M_{\text{r}}(\jmath \omega_0)}\nonumber\\
  &=(1+x^2-y^2+2x-4x_1) + \jmath(2y+2xy-4y_1).\label{eq:pfV_21}
\end{align}

In order to complete the proof, we show that $\operatorname{Re}\{G_+(\jmath\omega_0)\} > 1$, hence, we solve the following inequality
\begin{align}
  \left(1+x\right)
    -\sqrt{\frac{|z|+\operatorname{Re}\{z\}}{2}} > 2.\label{eq:pfV_22}
\end{align}
We simplify it as
\begin{align}
 & \left(2 (x-1)^2 - \operatorname{Re}\{z\} \right)^2 - \operatorname{Re}\{z\}^2 - \operatorname{Im}\{z\}^2 > 0,
\end{align}
substitute for $z$ into it from (\ref{eq:pfV_21}) and obtain
\begin{align}
 & - x_2(x_1-1)^2 - y_1 y_2(x_1+x_2-1) -2x_2^2 (x_1-1)  \nonumber \\
 & - x_2^3 - y_2^2(x_1+x_2) > 0. \label{eq:pfV_23}
\end{align}
We substitute from (\ref{eq:pfV_19a})-(\ref{eq:pfV_19d}) into (\ref{eq:pfV_23}) and get
\begin{align}
  l_{x,1} (k_{x,1}-1)^2 \omega_0^2 + \mathcal{O}(\omega_0^4, \omega_0^6, \omega_0^8, ...) > 0,\label{eq:pfV_24}
\end{align}
where $\mathcal{O}(\omega_0^4, \omega_0^6, \omega_0^8, ...)$ stands for the polynomial with terms $\omega_0^4$, $\omega_0^6$, $\omega_0^8$ etc. The lowest order term in (\ref{eq:pfV_24}) is $\omega_0^2$. Therefore, the inequality in (\ref{eq:pfV_22}) holds for $\omega_0$ close to zero if $l_{x,1} = 1/n_{\text{r},0}> 0$. By (\ref{eq:as_in_2c}) we assume that $n_{\text{r},0} > 0$, hence, $\operatorname{Re}\{G_+(\jmath\omega_0)\}>1$ and $||G_+||_{\infty}>1$. Similarly, it can be shown that $||G_-||_{\infty}>1$ if $0< \kappa <1$ and $n_{\text{f},0} > 0$. Hence, if $n_{\text{f},0} \neq n_{\text{r},0}$ then $\kappa \neq 1$ and either $||G_+(s)||_{\infty} > 1$ or $||G_-(s)||_{\infty} > 1$.
\end{proof}

\section{Implications for the graphs with asymmetric coupling}
\subsection{Path-graph topology}

In this section, we follow the argument given in the Introduction that certain features in the performance of the distributed system can be inferred from the analysis of the wave propagation between the agents because of the local nature of the AWTFs.

\begin{defn}
  We say that the distributed system is \emph{locally string stable} if the AWTFs are asymptotically stable and
  \begin{align}
    ||G_+(s)||_{\infty} \leq 1 \;\;\;\; \text{and} \;\;\;\; ||G_-(s)||_{\infty} \leq 1.
  \end{align}
  Otherwise, the system is called \emph{locally string unstable}.
\end{defn}

Similarly to the string stabilities mentioned in the introduction, the local string stability also deals with the performance of the distributed system. It also describes whether the disturbance acting on an agent amplifies as it propagates through the system. However, the local string stability describes the performance from the local point-of-view without considering the whole distributed system. The local description is particularly advantageous for a large distributed system, where the traditional Laplacian approach is difficult to apply.


The effect of the local string instability is illustrated in Fig.~\ref{fig:local_string_stability}. The left panels show the response of the system with the symmetric bidirectional control, which is locally string stable. The right panels show the response of the system with the asymmetric control, where the asymmetry is in both position and velocity. This makes the system locally string unstable. We can see (top-right panel) that the overshoot of the locally string unstable system increases with the index of the agent. The more agents the wave transmits through, the larger the overshoot is, due to the fact that $||G_{+}(s)||_{\infty} > 1$. This does not happen for the locally string stable system (top-left panel). We can see that even the locally string stable system eventually overshoots the input signal (bottom-left panel) which is due to the reflection of the wave on the last agent.



\begin{figure}[ht]
 \centering
  \includegraphics[width=0.49\textwidth]{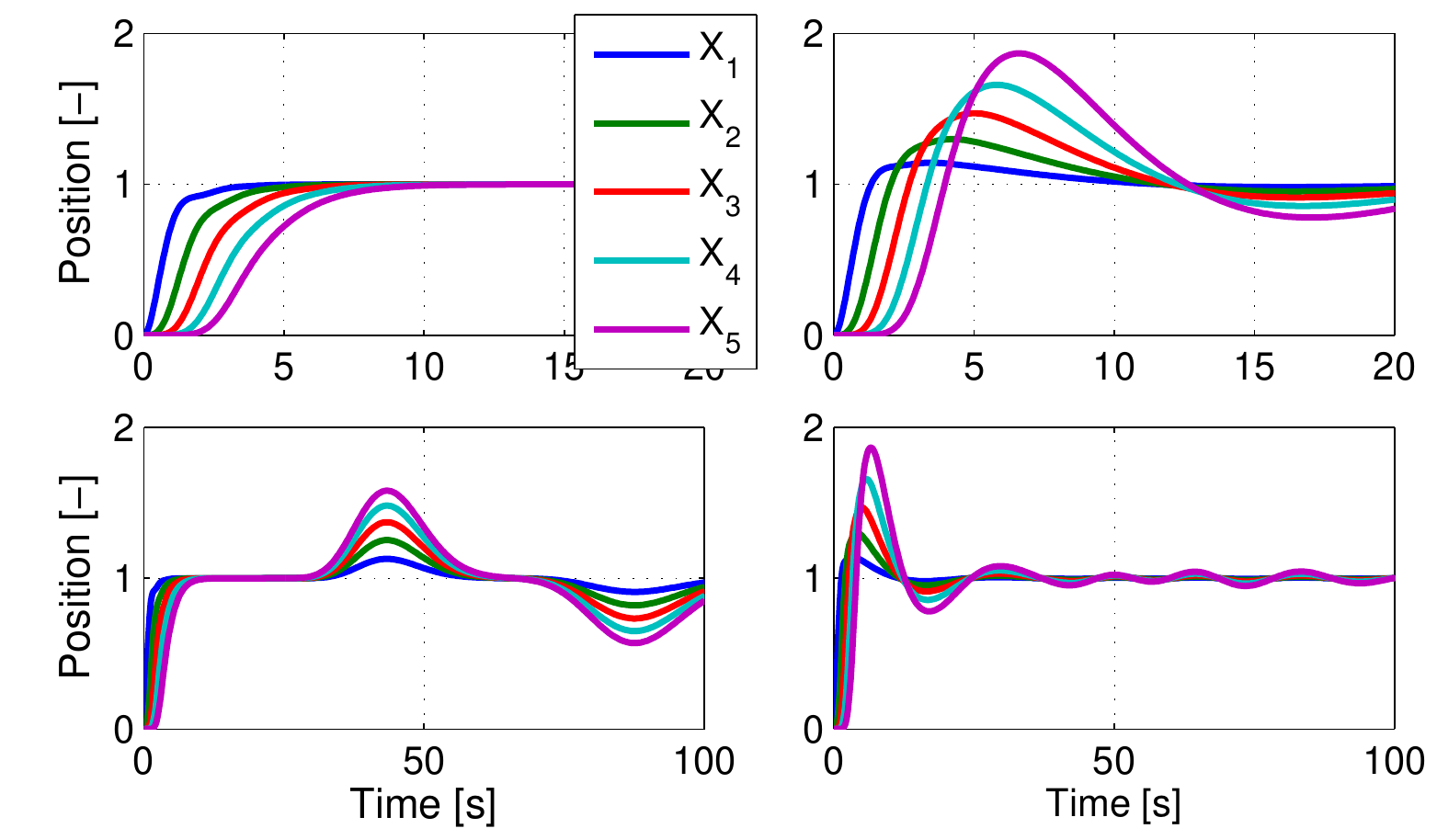}
  \caption{The comparison of the response of the locally string stable (left panels) with the locally string unstable system (right panels).}
  \label{fig:local_string_stability}
\end{figure}

The main result of the paper is given in the following Theorem.

\begin{theorem}\label{thm:sym_coupl}
If i) there are two integrators in the dynamics of the agents, ii) the AWTFs given by (\ref{eq:AWTF_1}) and (\ref{eq:AWTF_2}) are asymptotically stable, and iii) the positional coupling is asymmetric, then the distributed system {\color{\revA}with the path-graph topology defined in Definition~\ref{def:path_graph}} is locally string unstable.
\end{theorem}
\begin{proof}
If the positional coupling is asymmetric, then $n_{\text{f},0}/d_{\text{f},0} \neq n_{\text{r},0}/d_{\text{r},0}$ by Definition~\ref{def:pos_coupling}. Since we can always transform $M_{\text{f}}(s)$ and $M_{\text{r}}(s)$ such that $d_{\text{f},0} = d_{\text{r},0}=1$, then $n_{\text{f},0} \neq n_{\text{r},0}$. Therefore, $||G_{+}(s)||_{\infty} > 1$ or $||G_{-}(s)||_{\infty} > 1$, which follows from Theorem~\ref{thm:two_integrators_Mf_Mr}, and the distributed system is locally string unstable.

\end{proof}

We can interpret Theorem~\ref{thm:sym_coupl} as follows. The inequality $||G_{+}(s)||_{\infty} > 1$ causes that the disturbance is amplified as it propagates {\color{\revA}along a path graph} from $X_{i}(s)$ to $X_{i+1}(s)$, from $X_{i+1}(s)$ to $X_{i+2}(s)$, from $X_{i+2}(s)$ to $X_{i+3}(s)$ and so on. The larger the path graph is, the more the disturbance is amplified. Similarly, if $||G_{-}(s)||_{\infty} > 1$, then the disturbance is amplified as it propagates in the opposite direction.


Theorem~\ref{thm:sym_coupl} is in agreement with the results of \cite{Barooah2009}, \cite{Tangerman2012} or \cite{Herman2014c}, where it is stated that, if the asymmetry is in the form of $M_{\text{f}}(s) = \mu M_{\text{r}}(s)$ with $\mu$ being a constant gain, then the system with the path-graph topology is string unstable. However, Theorem~\ref{thm:sym_coupl} is more general since it states that the distributed system is string unstable if the DC gain of $M_{\text{f}}/M_{\text{r}}$ is not equal to one. Hence, it allows the asymmetry to be more complex.

We should emphasize that Theorem~\ref{thm:sym_coupl} does not disprove an asymmetry in the velocity coupling. In fact, the asymmetric velocity coupling may improve the transient of the system, as shows the simulation example in Section~\ref{sec:simulations}.





{\color{\revA}
\section{Implications on the string stability}

This section shows the relation between the local string stability and the string stability. Although there are several definitions of the string stability, we believe that the following definition captures the essence of the string stability. It is used for example in \cite{Ploeg2014}, \cite{Seiler2004a} or \cite{Eyre1998a}.
\begin{defn} \label{def:string_stability_classic}
  A distributed system with the path-graph topology is called string stable if there are upper bounds on the $H_{\infty}$ norms of the transfer functions $\Gamma_i(s)$ that do not depend on the number of agents. $\Gamma_i(s)$ are the transfer functions from the disturbance that acts on the input of arbitrarily agent to the output of another arbitrarily agent of the system.
\end{defn}


The model (\ref{eq:eq1}) of $n$th agent with the disturbance $\Delta_n(s)$ applied to it is given by
\begin{align}
    X_{n}(s) = &M_{\text{f}}(s)(X_{n-1}(s)-X_{n}(s)+\Delta_n(s)) \nonumber\\
  &+ M_{\text{r}}(s)(X_{n+1}(s)-X_{n}(s)).\label{eq:eq_disturbance_n}
\end{align}


\begin{lemma}\label{lemma:connection_to_string_stability}
  If the conditions in Theorem~\ref{thm:sym_coupl} hold and the distributed system with the path-graph topology is locally string unstable, then it is also string unstable in the sense of Definition~\ref{def:string_stability_classic}.
\end{lemma}
\begin{proof}
The transfer function from $\Delta_1(s)$ to $X_N(s)$ is the same as the transfer function from $X_0(s)$ to $X_N(s)$. The way how to find this transfer function in terms of the AWTFs is shown in Section~\ref{sec:wave_approach_discussion}. It is given as
  \begin{align}
   \frac{X_N(s)}{\Delta_1(s)} = \frac{G_+^{N}(s) + T_N(s)G_+^{N}(s)}{1-T_N(s)G_-^{N-1}(s)T_1(s)G_+^{N-1}(s)}.\label{eq:proof_tf_dist_1}
  \end{align}
Analogously, the transfer function from $\Delta_N(s)$ to $X_1(s)$ is
\begin{align}
   \frac{X_1(s)}{\Delta_N(s)} = \frac{G_-^{N-1}(s)G_{\Delta}(s) + T_1(s)G_-^{N-1}G_{\Delta}(s)}{1-T_N(s)G_-^{N-1}(s)T_1(s)G_+^{N-1}(s)},\label{eq:proof_tf_dist_2}
\end{align}
where $G_{\Delta}(s) = B_N(s)/\Delta_N(s)$. The exact value of $G_{\Delta}(s)$ is redundant for the proof.

The proof of Theorem~\ref{thm:two_integrators_Mf_Mr} shows that the asymmetric positional coupling causes that either $|G_+(\jmath\omega_0)|>1$ or $|G_-(\jmath\omega_0)|>1$ for $\omega_0 \rightarrow 0$. However, we can see from the DC gain analysis in Lemma~\ref{lemma:DC_gains} that $|G_+(\jmath\omega_0)G_-(\jmath\omega_0)|<1$.

The substitution of $|G_+(\jmath\omega_0)|>1$ into (\ref{eq:proof_tf_dist_1}) for $N\rightarrow \infty$ gives
\begin{align}
 \lim_{N \rightarrow \infty, \omega_0 \rightarrow 0} \left|\frac{X_N(\jmath\omega_0)}{\Delta_1(\jmath\omega_0)}\right|  = \infty.
\end{align}
We can see that the $H_{\infty}$ norm of the transfer function $X_N(s)/\Delta_1(s)$ is not bounded in the number of agents if $||G_+(s)||_{\infty}>1$. Analogously, if $||G_-(s)||_{\infty}>1$ then the $H_{\infty}$ norm of the transfer function $X_1(s)/\Delta_N(s)$ is not bounded.

\end{proof}

The combination of Theorem~\ref{thm:sym_coupl} and Lemma~\ref{lemma:connection_to_string_stability} gives the following Corollary.

\begin{corollary}
  The symmetric positional coupling is a necessary (but not sufficient) condition for the string stability of the system with the path-graph topology.
\end{corollary}}


{\color{\revB}
\section{Constant-time-headway spacing policy}

This section briefly discusses the applicability of the transfer function approach to the constant-time-headway spacing policy, see for instance \cite{Eyre1998a}, described by
\begin{align}
U_n(s) =& C_{\text{r}}(s)(X_{n-1}(s)-(1+hs)X_n(s))\nonumber\\
&+C_{\text{f}}(s)(X_{n+1}(s)-(1+hs)X_n(s)),
\end{align}
where $h$, a constant-headway time between the agents, is to be kept fixed.

The Asymmetric Wave Transfer Functions for the constant-time-headway control are given by the following lemma.
\begin{lemma}
  \begin{align}
  G_{+,\textnormal{H}}(s) &= \frac{1}{2}\beta_{\textnormal{H}}(s) - \frac{1}{2}\sqrt{\beta_{\textnormal{H}}^2(s) - 4\frac{M_{\textnormal{f}}(s)}{M_{\textnormal{r}}(s)}},\label{eq:AWTF_time_headway_1}\\
  G_{-,\textnormal{H}}(s) &= \frac{1}{2}\alpha_{\textnormal{H}}(s) - \frac{1}{2}\sqrt{\alpha_{\textnormal{H}}^2(s) - 4\frac{M_{\textnormal{r}}(s)}{M_{\textnormal{f}}(s)}},\label{eq:AWTF_time_headway_2}
\end{align}
where
\begin{align}
  \beta_{\textnormal{H}}(s) &= \frac{1+(1+hs)M_{\textnormal{f}}(s)+(1+hs)M_{\textnormal{r}}(s)}{M_{\textnormal{r}}(s)},\\
  \alpha_{\textnormal{H}}(s) &= \frac{1+(1+hs)M_{\textnormal{r}}(s)+(1+hs)M_{\textnormal{f}}(s)}{M_{\textnormal{f}}(s)}.
\end{align}
\end{lemma}
\begin{proof}
  The derivation (\ref{eq:AWTF_time_headway_1}) and (\ref{eq:AWTF_time_headway_2}) is analogous to that in Lemma~\ref{lem:AWTF}. In this case (\ref{eq:lem1_AWTF_pf1}) changes to
  \begin{align}
     A_n + B_n &= M_{\text{f}}\left((G_+^{-1}A_{n}+G_-B_{n})-(1+hs)(A_n+B_n)\right)\nonumber\\
  & + M_{\text{r}}\left((G_+A_{n}+G_-^{-1}B_{n})-(1+hs)(A_n+B_n)\right).
  \end{align}
  The decomposition into $A$ and $B$ parts yields
  \begin{align}
  &G_{+,\text{H}}^{2}-\beta_{\text{H}} G_{+,\text{H}} + \frac{M_{\text{f}}}{M_{\text{r}}} = 0,\\
  &G_{-,\text{H}}^{2}-\alpha_{\text{H}} G_{-,\text{H}} + \frac{M_{\text{r}}}{M_{\text{f}}} = 0.
\end{align}
The only proper solutions of these quadratic equations are (\ref{eq:AWTF_time_headway_1}) and (\ref{eq:AWTF_time_headway_2}).
\end{proof}

The local string stability analysis can also be carried out for a constant-time-headway policy. However, we have to emphasize that the proof of Theorem~\ref{thm:two_integrators_Mf_Mr} can not be directly used in this case due to the parameter $h$, but (\ref{eq:pfV_24}) must be changed to
\begin{align}
  (l_{x,1} &(k_{x,1}-1)^2 + hk_{y,1}(1-k_{x,1})-h^2k_{x,1}^2) \omega_0^2 \nonumber\\
  &+ \mathcal{O}(\omega_0^3, \omega_0^4, \omega_0^5, ...) > 0,
\end{align}
where $l_{x,1} = 1/n_{\text{r},0} > 0$, $k_{x,1} > 1$ and $k_{y,1} = (n_{\text{f},1}n_{\text{r},0} - n_{\text{f},0}n_{\text{r},1} - d_{\text{f},1}n_{\text{f},0}n_{\text{r},0} + d_{\text{r},1}n_{\text{f},0}n_{\text{r},0})/n_{\text{r},0}^2$. We can see that the dominant term
\begin{align}
  \left(l_{x,1} (k_{x,1}-1)^2 + hk_{y,1}(1-k_{x,1})-h^2k_{x,1}^2\right)
\end{align}
becomes negative for certain values of $h$, $k_{x,1}$ and $k_{y,1}$, which disproves that the real value of $G_{+,\text{H}}(s)$ is greater than one.

In fact, numerical simulations, which are not included in the paper, indicate that Theorem~\ref{thm:two_integrators_Mf_Mr} does not hold for the constant-time-headway policy. This is in agreement with the fact that the string stability can usually be corrected by the constant-time-headway policy, see Section~3.4~of~\cite{Eyre1998a}.
}

\section{Mathematical simulations}\label{sec:simulations}


The mathematical simulations compare three different control strategies for two different path graphs. The results are shown in Fig.~\ref{fig:sym_asym_comp}, where the agent is modelled as a double integrator with a linear model of friction controlled by a PI controller, that is
\begin{align}
  M_{\text{f}} = \frac{1}{3}\frac{4s+4}{s^2(s/3+1)}\label{eq:mat_sim_1}
\end{align}
for all three cases. $M_{\text{r}}$ are
\begin{align}
\addtocounter{equation}{1}
  M_{\text{r}} = M_{\text{f}},\;\;
  M_{\text{r}} = \frac{2.5}{4}M_{\text{f}},\;\;
  M_{\text{r}} = \frac{1}{3} \frac{2.5s+4}{s^2(s/3+1)},
  \tag{\arabic{equation} a, b, c}
  \label{eq:mat_sim_2}
\end{align}
for the left, middle and right panels, respectively.
\newcounter{multi_eq1}
\setcounter{multi_eq1}{\value{equation}}

\begin{figure*}[ht]
 \centering
  \includegraphics[width=0.95\textwidth]{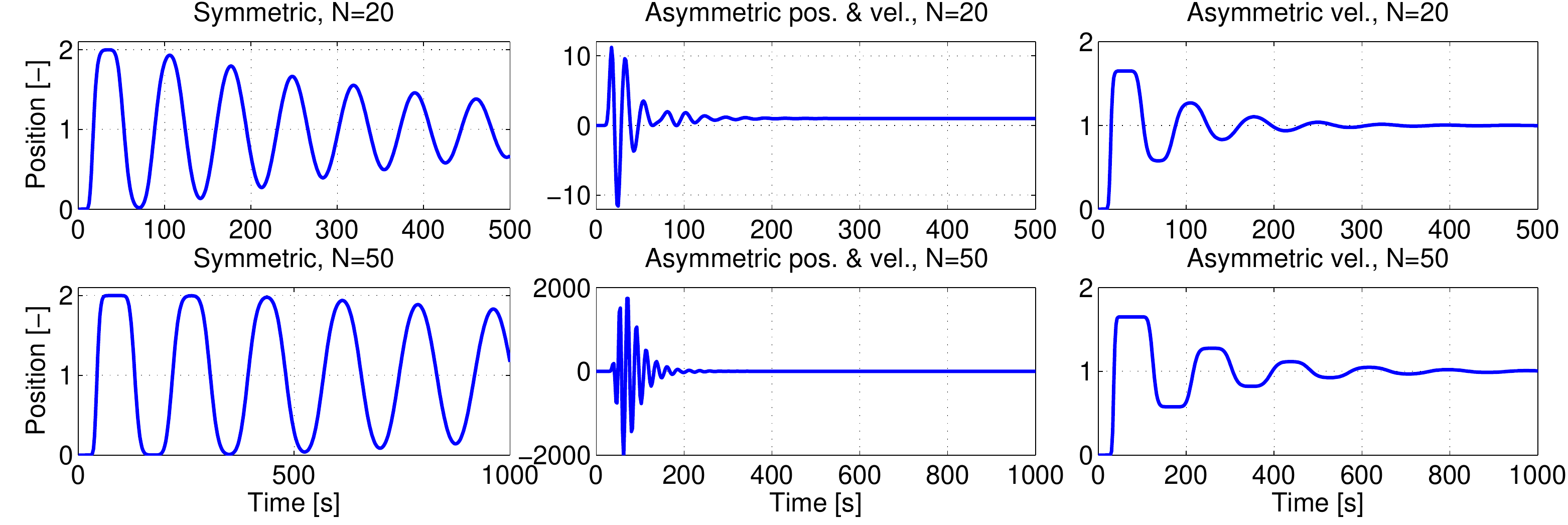}
  \caption{
  The numerical simulations showing the position of the last agent in the distributed system with path-graph topology when the leader changes its position from $0$ to $1$. The figure compares three different bidirectional control strategies: i) the symmetric (the left panels) defined by (\arabic{multi_eq1}a), ii) the traditional asymmetric control with asymmetries in both positional and velocity couplings (the middle panels), see (\arabic{multi_eq1}b), and iii) the combined symmetric positional with asymmetric velocity couplings (the right panels), see (\arabic{multi_eq1}c). The top and bottom panels show the system with 20 and 50 agents, respectively.}
  \label{fig:sym_asym_comp}
\end{figure*}

We can see that the symmetric bidirectional control has a very long transient (the left panel), which is shortened when the asymmetry is introduced to both positional and velocity couplings (the middle panel). However, the asymmetry in the positional coupling causes a large overshoot that even scales with the size of the graph, which is due to the local string instability. When the positional coupling is kept symmetric and the velocity coupling asymmetric (the right panel), the overshoot is smaller than for the symmetric case. Moreover, we can see that the transient are scaled approximately linearly with the size of the graph.


An independent validation of the AWTF approach is shown in Fig.~\ref{fig:wave_state_comp}. The numerical simulation shows the response of 20 agents with the path-graph topology, where
\begin{align}
  M_{\text{f}}(s) = \frac{1}{3}\frac{4s+4}{s^2(s/3+1)}, \;\;\;\; M_{\text{r}}(s) = \frac{1}{3}\frac{2.5s+4}{s^2(s/3+1)}.
\end{align}
We can see excellent agreement between the state-space approach based on (\ref{eq:eq1}) and the AWTF approach. The waves $A_{10}$ and $B_{10}$ are computed by (\ref{eq:com_1a}) and (\ref{eq:com_1b}). The wave $X_{10}$ is a sum of $A_{10}$ and $B_{10}$ and can be alternatively computed by (\ref{eq:com_2}). We can also see that the approximation $x_{10}(t) \approx a_{10}(t)$ holds in the time-domain for the first 20 seconds, then the wave returns back to the $10$th agent which causes an increase of $B_{10}$.

\begin{figure}[ht]
 \centering
  \includegraphics[width=0.49\textwidth]{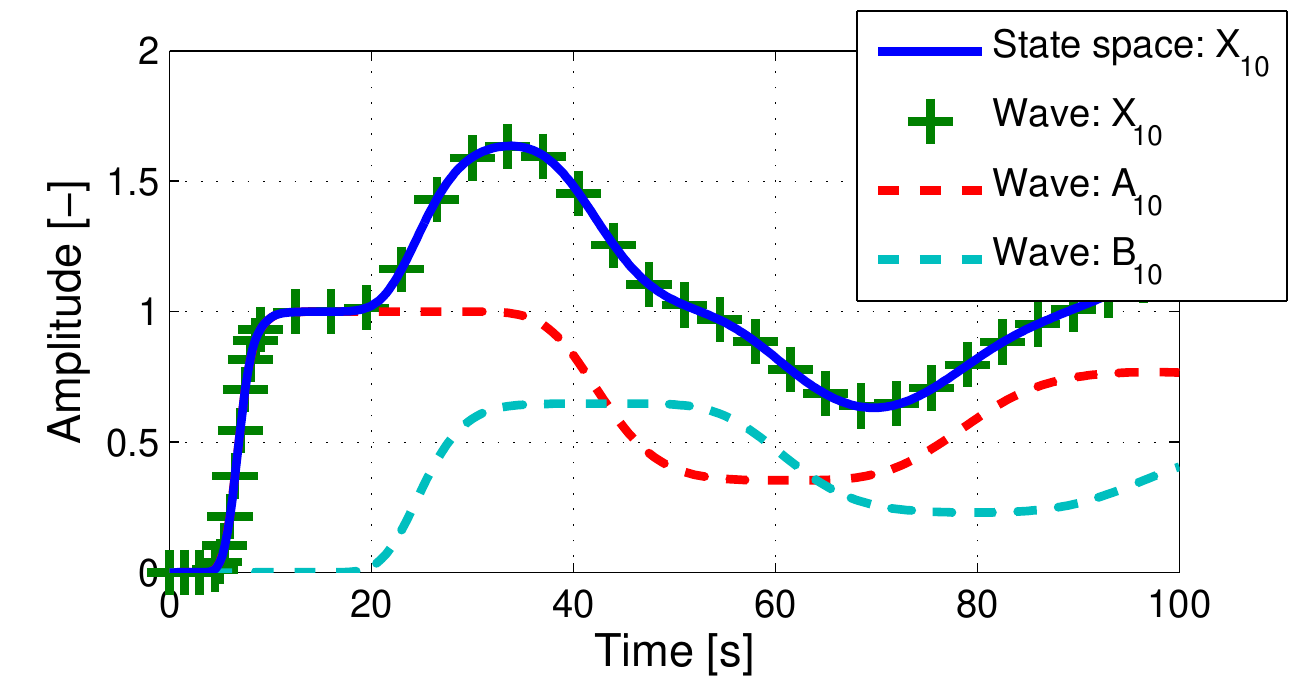}
  \caption{
  The comparison of the positions of the $10$th agent in the system considered at the top right panel in Fig.~\ref{fig:sym_asym_comp} simulated by the state-space approach using (\ref{eq:eq1}) (blue solid line) and by the AWTF's approach using (\ref{eq:anp_s}), (\ref{eq:bnp_s}) and Lemmas \ref{lem:AWTF_stab} and \ref{lem:refl} (green crosses). The two components $A_{10}$ and $B_{10}$ from (\ref{eq:pos_decomp}) are shown with the dashed red and blue lines, respectively. The response on the step change of $X_0$ is shown.}
  \label{fig:wave_state_comp}
\end{figure}

Fig.~\ref{fig:bode_comp} shows the numerical validation of Theorem~\ref{thm:two_integrators_Mf_Mr} for $M_{\text{f}}$ and $M_{\text{r}}$ defined by (\ref{eq:mat_sim_1}) and (\arabic{multi_eq1}b)-(\arabic{multi_eq1}c). We can see that, if there is asymmetry in the positional coupling (solid line), i.e. $n_{\text{f},0} \neq n_{\text{r},0}$, then the $\mathcal{H}_{\infty}$ norm of $G_+$ is greater than one. The norm is reduced to one by making the positional coupling symmetric (dashed line).
\begin{figure}[ht]
 \centering
  \includegraphics[width=0.48\textwidth]{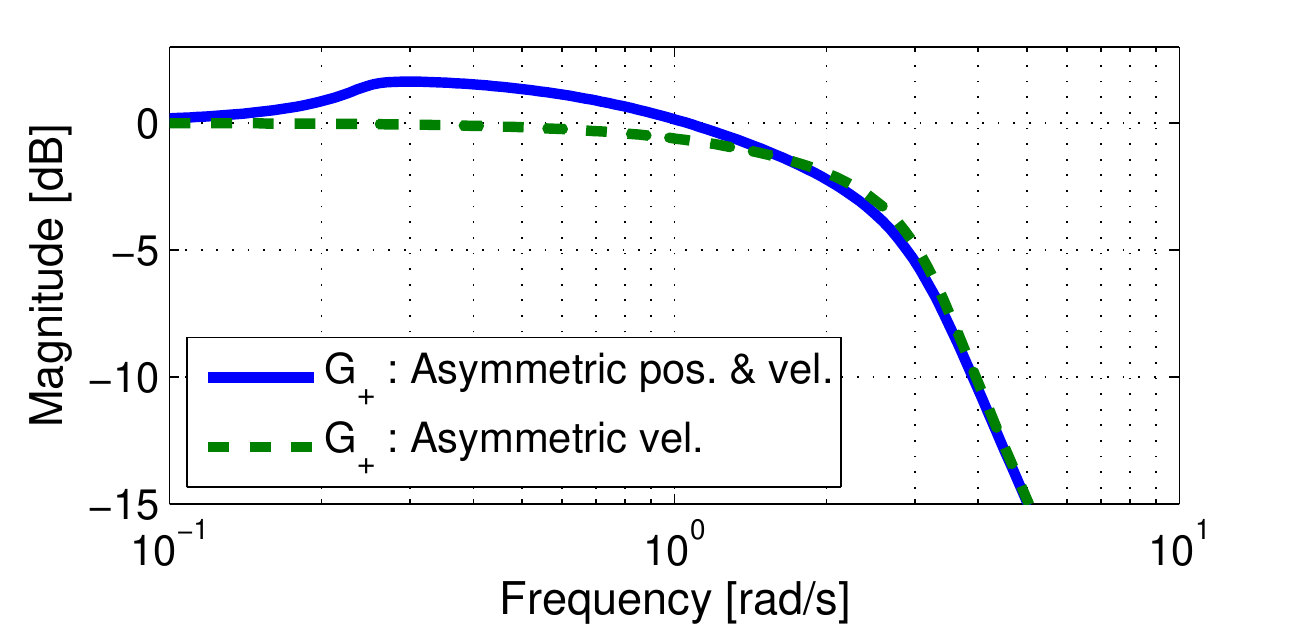}
  \caption{
  The comparison of the frequency characteristics for two different transfer functions $G_+(s)$. The asymmetries in the couplings are defined as in Fig.~\ref{fig:sym_asym_comp}.}
  \label{fig:bode_comp}
\end{figure}

{\color{\revA}
\section{Generalized path-graph topology}
}
The key part in the derivation of $G_{+}$ and $G_{-}$ is that the agent has exactly two neighbours. Therefore, we can apply the decomposition of the waves from (\ref{eq:pos_decomp})-(\ref{eq:AWTF_2}) to each agent that has two neighbours. In other words, if the agents are connected in {\color{\revA}a generalized path-graph topology, see Definition~\ref{def:generalized_path_graph}}, then $G_+$ and $G_-$ are the same.


{\color{\revA}
\begin{defn}\label{def:generalized_path_graph}
The \emph{generalized path graph} is defined as the path graph in Definition~\ref{def:path_graph}, except that the $N$th agent has more than two neighbours.
\end{defn}
Obviously, i) the condition (\ref{eq:eq2}) for the $N$th agent of the generalized path graph is relaxed, and ii) the system containing the generalized path graph has more than $N+1$ agents. Figure~\ref{fig:graph_topologies} shows the examples of a path graph and generalized path graphs.
}

\begin{figure}[ht]
 \centering
  \includegraphics[width=0.49\textwidth]{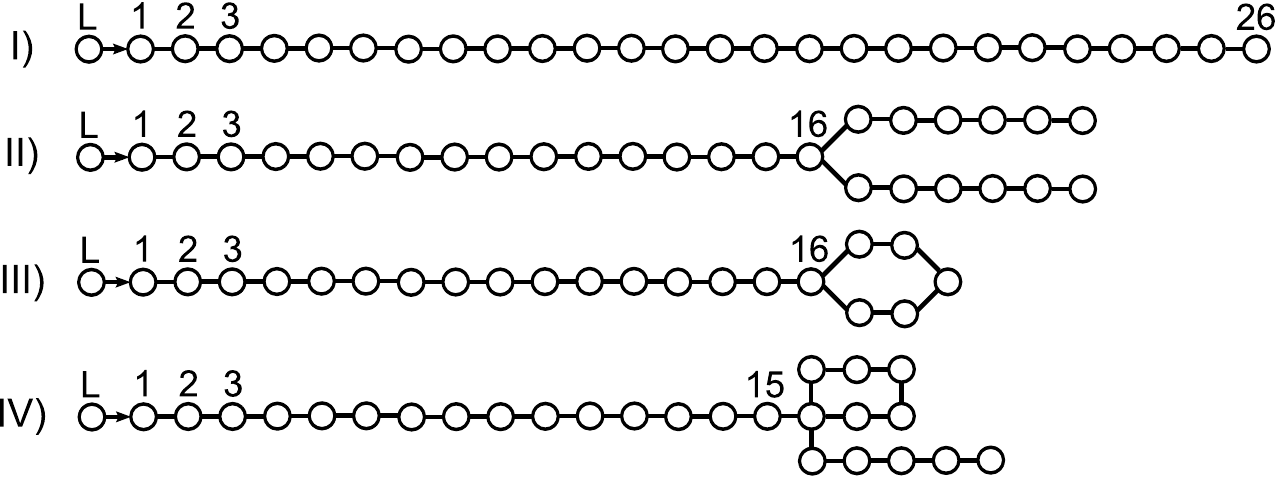}
 \caption{{\color{\revA}I) A path graph with $N=27$, and II)-IV) generalized path graphs with $N=17$.}}
  \label{fig:graph_topologies}
\end{figure}

{\color{\revA}The approximations (\ref{eq:com_3}) and (\ref{eq:com_4}) are based on the fact that the boundary effect of the $N$th agent does not influence the travelling wave from the leader for a certain initial time. Since the generalized path-graph topology modifies only the boundary condition on the $N$th agent, the same approximations (\ref{eq:com_3}) and (\ref{eq:com_4}) hold for both types of the graphs. Moreover,
\begin{theorem}\label{thm:sym_coupl_generalized}
  Theorem~\ref{thm:sym_coupl} holds for a generalized path graph.
\end{theorem}
\begin{proof}
The proof of the local string instability is the same as that of Theorem~\ref{thm:sym_coupl}, since the proof is independent of the boundary condition (\ref{eq:eq2}).
\end{proof}}

\subsection{Mathematical simulations}
The Remark after Lemma~\ref{lem:AWTF} states that the boundary agent changes the way of how the wave is reflected from boundary but it does not affect the wave travelling towards or from it. We demonstrate this feature on the systems with topologies from Fig.~\ref{fig:graph_topologies}. We assume that all the agents have the identical dynamics defined by (\ref{eq:mat_sim_1}) and (\arabic{multi_eq1}a). We can see in Fig.~\ref{fig:boundary_effect_comp} {\color{\revA}and Fig.~\ref{fig:boundary_effect_comp_unstable}} that the responses of the agents are identical in all systems for about first 18 seconds, which corresponds to the time needed for the wave to travel from the leader to the boundary agent. After that, the wave is reflected from the boundary agent and reaches the agents again. In the case of the top-left panel, the wave reflects later.

{\color{\revA}We can also see that the numerical simulations confirms Theorem~\ref{thm:sym_coupl_generalized}. The initial responses of the systems with the symmetric positional coupling in Fig.~\ref{fig:boundary_effect_comp} do not amplify the disturbance. On the other hand, if the condition on the symmetric positional coupling is violated as in Fig.~\ref{fig:boundary_effect_comp_unstable}, then the disturbance is amplified as it propagates along the generalized path graph, which is predicted by Theorem~\ref{thm:sym_coupl_generalized}.

This example also shows us the main advantage of the wave transfer function approach and the local string stability concept. We can determine if the disturbance is amplified as it travels in the generalized path graph without the necessity to analyze the response of the entire distributed system. For instance, we can say that, if the agents do not have symmetric positional coupling, then the disturbance in the generalized path graph is amplified regardless of the interaction of the $N$th agent with the other parts of the system.}



\begin{figure}[ht]
 \centering
  \includegraphics[width=0.49\textwidth]{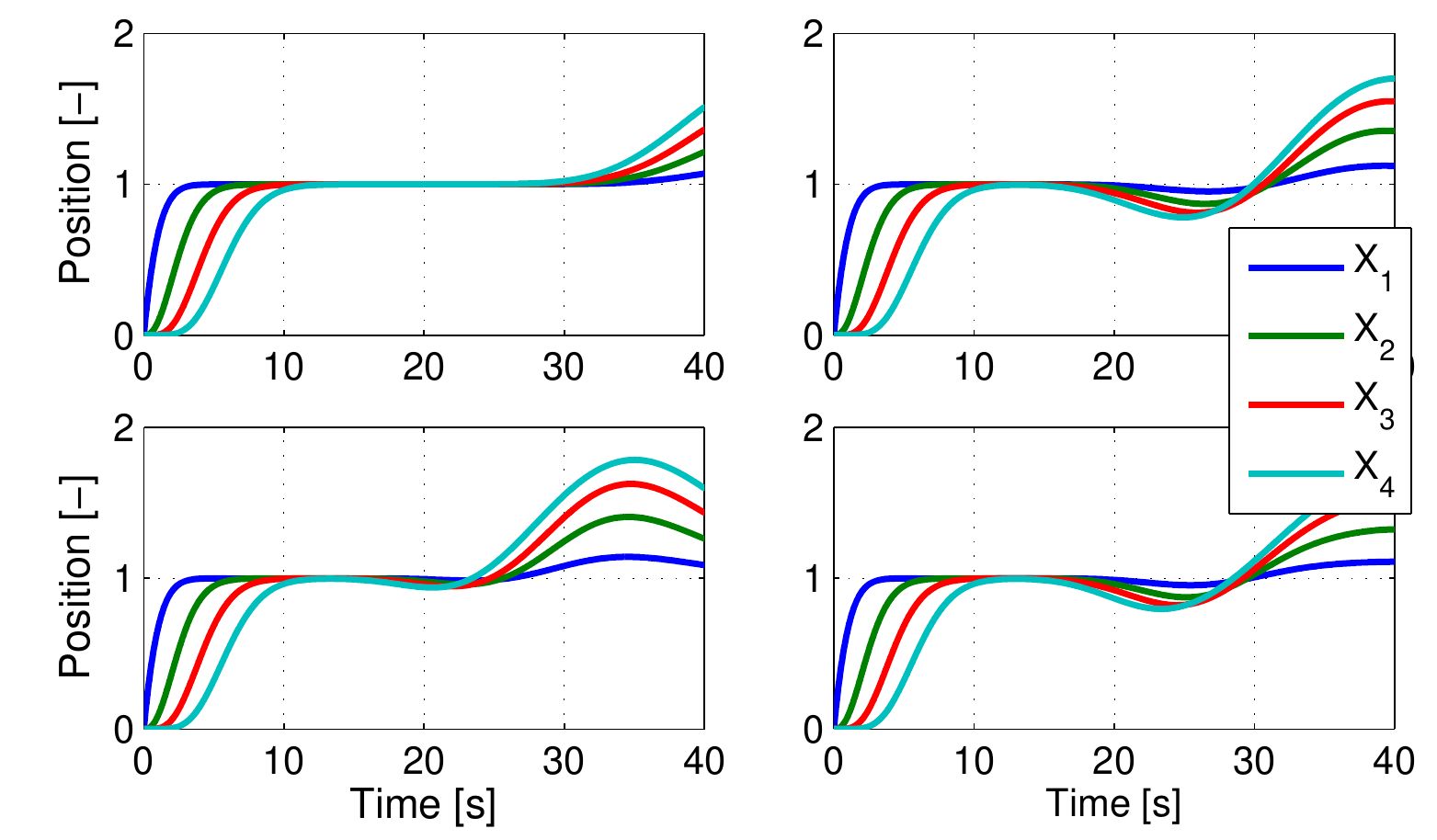}
  \caption{The comparison of the responses of the systems with four different interaction topologies. The topologies are given in Fig.~\ref{fig:graph_topologies} and they are related to this figures as follows: top-left panel - topology I); top-right panel - topology II); bottom-left panel - topology III); bottom-right panel - topology IV).}
  \label{fig:boundary_effect_comp}
\end{figure}

\begin{figure}[ht]
 \centering
  \includegraphics[width=0.49\textwidth]{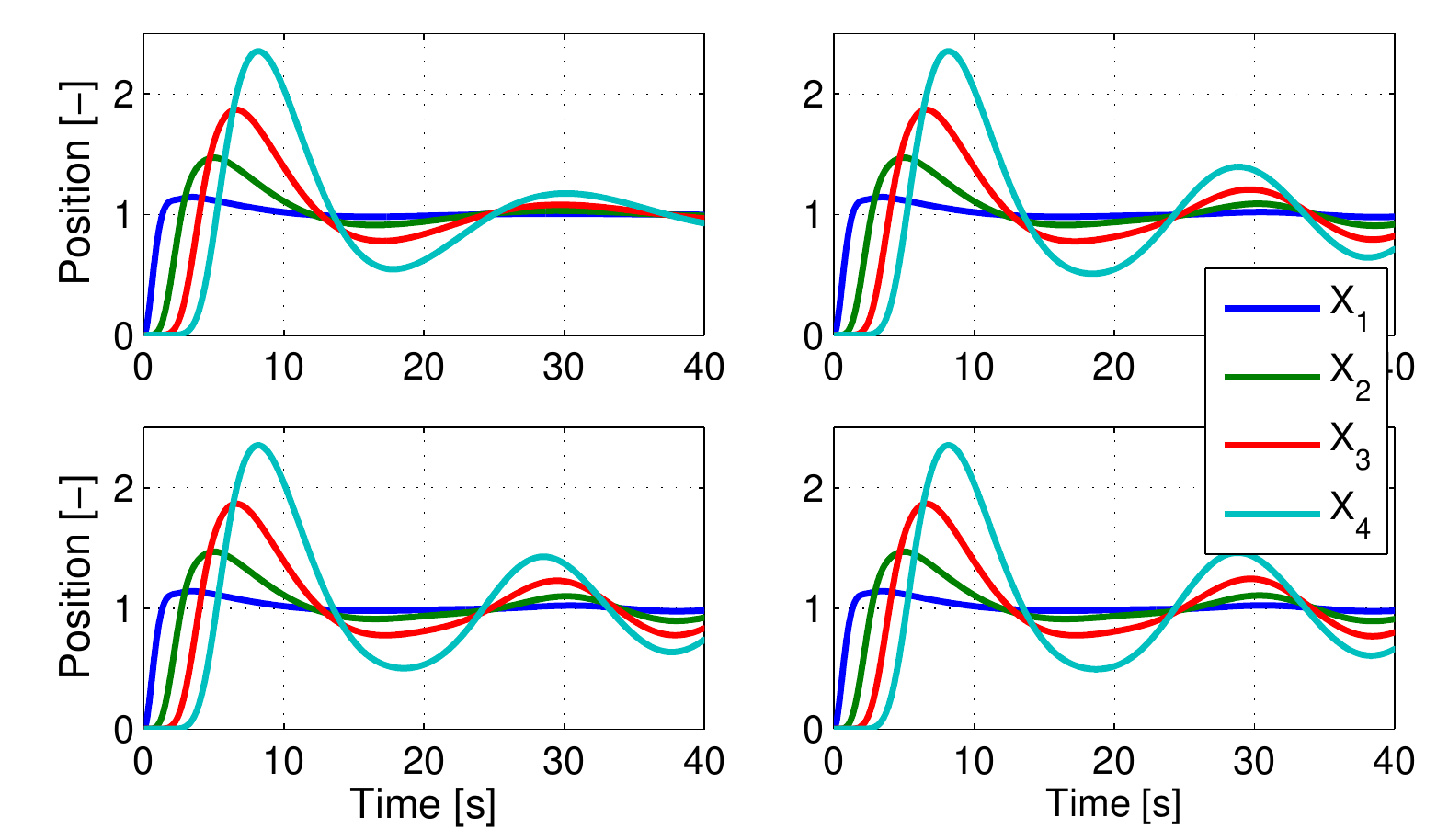}
  \caption{{\color{\revA}The same as in Fig~\ref{fig:boundary_effect_comp} but for the agents defined by (\ref{eq:mat_sim_1}) and (\arabic{multi_eq1}b).}}
  \label{fig:boundary_effect_comp_unstable}
\end{figure}


\section{Conclusions}
The paper examined a distributed system with constant-spacing policy and asymmetric bidirectional control, where the coupling between the agents is allowed to be arbitrarily complex. The proposed approach reveals that the symmetric positional coupling, i.e. identical DC gains of the controllers, is necessary for the string stability of the distributed system. This finding does not disprove the asymmetry for other couplings. In fact, it is numerically shown that, if the asymmetry in the velocity coupling is adjusted properly, then the system's performance may be improved.





\bibliographystyle{IEEEtran}
\bibliography{2015-Path_graph_asymmetry}

\end{document}